\documentclass[reqno,11pt,
a4paper]{amsart}
\allowdisplaybreaks[4]

\usepackage{amssymb,stmaryrd}
\usepackage{amsmath}
\usepackage{color}
\usepackage
{hyperref}
\usepackage[all]{xy}
\usepackage{url}

\def\d{\mathrm{d}}

\newcommand*{\eval}[1]{\left.#1\right|}

\newcommand*{\sd}[2]{\{\,#1\mid#2\,\}}

\newcommand*{\pd}[2]{\mathchoice{\frac{\partial#1}{\partial#2}}
  {\partial#1/\partial#2}{\partial#1/\partial#2}
  {\partial#1/\partial#2}}

\renewcommand*{\d}{\mathinner{\!}\mathrm{d}}
\newcommand{\In}{\mathrm{i}}
\newcommand{\Ld}{\mathrm{L}}

\let\epsilon=\varepsilon
\let\phi=\varphi

\DeclareMathOperator{\sym}{sym}

\DeclareMathOperator{\cl}{cl}
\DeclareMathOperator{\ic}{\mathbf{int}}

\theoremstyle{theorem}
\newtheorem{theorem}{Theorem}
\newtheorem{proposition}{Proposition}
\newtheorem{corollary}{Corollary}
\theoremstyle{definition}
\newtheorem{definition}{Definition}

\newtheorem*{example*}{Example}
\newtheorem{example}{Example}
\theoremstyle{remark}
\newtheorem{remark}{Remark}
\newtheorem*{remark*}{\textbf{Remark}}

\newcommand{\cprime}{\/{\mathsurround=0pt$'$}}

\begin{document}

\title{Integrability in differential coverings}

\author[Joseph~Krasil{\cprime}shchik]{Joseph
  Krasil{\cprime}shchik\email{josephkra@gmail.com}\address{Slezsk\'{a}
    univerzita v Opav\v{e}, Matematick\'{y} \'{u}stav v Opav\v{e}, Na
    Rybn\'{\i}\v{c}ku 626/1, 746 01 Opava, Czech Republic}\address{Independent
    University of Moscow, 119002, B.~Vlasyevskiy Per. 11, Moscow,
    Russia}}\thanks{I am grateful to the Mathematical Institute of the
  Silesian University in Opava for support and comfortable working condition.}


\keywords{Geometry of differential equations, integrability, symmetries,
  conservation laws, differential coverings}

\subjclass[2010]{37K05, 37K10, 37K35}

\begin{abstract}
  Let~$\tau\colon\tilde{\mathcal{E}}\to\mathcal{E}$ be a differential covering
  of a PDE~$\tilde{\mathcal{E}}$ over~$\mathcal{E}$. We prove that
  if~$\mathcal{E}$ possesses infinite number of symmetries and/or conservation
  laws then~$\tilde{\mathcal{E}}$ has similar properties.
\end{abstract}
\maketitle

\section*{Introduction}
\label{sec:introduction}

The notion of a covering (or, better, \emph{differential} covering) was
introduced by A.~Vinogradov in~\cite{VinCat} and elaborated in detail later
in~\cite{VinKrasAdd} and~\cite{VinKrasTrends}. Coverings, explicitly or
implicitly, provide an adequate background to deal with nonlocal aspects in
the geometry of PDEs (nonlocal symmetries and conservation laws,
Wahlquist-Estabrook prolongation structures, Lax pairs, zero-curvature
representations, etc.). Coverings of a special type (the so-called tangent and
cotangent one) are efficient in analysis and construction of Hamiltonian
structures and recursion operators, see~\cite{KrasVer}. A very interesting
development in the theory of coverings can also be found in~\cite{Igonin-GPh}.

In this paper, we solve the following naturally arising problem: let a
covering~$\tau\colon\tilde{\mathcal{E}}\to\mathcal{E}$ be given and the
equation~$\mathcal{E}$ is known to possess infinite number of symmetries
and/or conservation laws. Is~$\tilde{\mathcal{E}}$ endowed with similar
properties? The answer, under reasonable assumptions, is positive.

In Section~\ref{sec:basic-notions}, we present a short introduction to the
theory of coverings based mainly on~\cite{VinKrasTrends} and formulate and
prove necessary auxiliary facts. Section~\ref{sec:main-result} contains the
proof of the main result for the case of Abelian coverings. Finally, the
non-Abelian case is discussed in Section~\ref{sec:disc-non-abel}.

\section{Basic notions}
\label{sec:basic-notions}

For a detailed exposition of the geometrical approach to PDEs we refer the
reader to the books~\cite{KLV} and~\cite{AMS}. Coverings are discussed
in~\cite{VinKrasTrends}.

\subsection*{Equations}
\label{sec:equations}

Let~$\pi\colon E\to M$, $\dim M=n$, $\dim E=m+n$, be a locally trivial vector
bundle and~$\mathcal{E}\subset J^\infty(\pi)$ be an infinitely prolonged
differential equation embedded to the space of infinite jets. One has the
surjection~$\pi_\infty\colon\mathcal{E}\to M$. The main geometric structure
on~$\mathcal{E}$ is the \emph{Cartan connection}~$\mathcal{C}\colon
Z\mapsto\mathcal{C}_Z$ that takes vector fields on~$M$ to those
on~$\mathcal{E}$. Vector fields of the form~$\mathcal{C}_Z$ are called
\emph{Cartan fields}. The connection is flat,
i.e.,~$\mathcal{C}_{[Z,Y]}=[\mathcal{C}_Z,\mathcal{C}_Y]$ for any vector
fields on~$M$. The corresponding horizontal distribution (the \emph{Cartan
  distribution}) on~$\mathcal{E}$ is integrable and its maximal integral
manifolds are solutions of~$\mathcal{E}$.  We always assume~$\mathcal{E}$ to
be \emph{differentially connected} which means that for any set of linearly
independent vector fields~$Z_1,\dots,Z_n$ on~$M$ the system
\begin{equation*}
  \mathcal{C}_{Z_i}(h)=0,\qquad i=1,\dots,n,
\end{equation*}
has constant solutions only.

If~$x^1,\dots,x^n$ are local coordinates on~$M$ then the Cartan connection
takes the partial derivatives~$\pd{}{x^i}$ to the \emph{total
  derivatives}~$D_{x^i}$ on~$\mathcal{E}$. Flatness of~$\mathcal{C}$ amounts
to the fact that the total derivatives pair-wise
commute,~$[D_{x^i},D_{x^j}]=0$.

A $\pi_\infty$-vertical vector field~$S$ is a \emph{symmetry} of~$\mathcal{E}$
if it commutes with all Cartan fields, i.e.,~$[S,\mathcal{C}_Z]=0$ for
all~$X$. The set of symmetries is a Lie algebra over~$\mathbb{R}$ denoted
by~$\sym\mathcal{E}$.

A differential $q$-form~$\omega$ on~$\mathcal{E}$, $q=0,1,\dots,n$, is
\emph{horizontal} if~$\In_V\omega=0$ for any $\pi_\infty$-vertical
field~$V$. The space of these forms is denoted
by~$\Lambda_h^q(\mathcal{E})$. Locally, horizontal forms are
\begin{equation*}
  \omega=\sum a_{i_1,\dots,i_q}\d x^{i_1}\wedge\dots\wedge\d x^{i_q},\qquad
  a_{i_1,\dots,i_q}\in\mathcal{F}(\mathcal{E}).
\end{equation*}
The \emph{horizontal de~Rham differential}~$\d_h\colon
\Lambda_h^q(\mathcal{E})\to \Lambda_h^{q+1}(\mathcal{E})$ is defined, whose
action locally is presented by
\begin{equation*}
  \d_h(a_{i_1,\dots,i_q}\d x^{i_1}\wedge\dots\wedge\d x^{i_q}) = \sum_{i=1}^n
  D_{x^i}(a_{i_1,\dots,i_q})\d x^i\wedge\d x^{i_1}\wedge\dots\wedge\d x^{i_q}.
\end{equation*}
A closed horizontal $(n-1)$-form is called a \emph{conservation law}
of~$\mathcal{E}$. Thus, conservation laws are defined by~$\d_h\omega=0$,
$\omega\in\Lambda_h^{n-1}(\mathcal{E})$. A conservation law is \emph{trivial}
if~$\omega=\d_h\rho$ for some~$\rho\in\Lambda_h^{n-2}(\mathcal{E})$. The
quotient space of all conservation laws modulo trivial ones is denoted
by~$\cl\mathcal{E}$.

If~$S\in\sym\mathcal{E}$ and~$\omega$ is a conservation law then the Lie
derivative~$\Ld_S\omega$ is a conservation law as well and trivial
conservation laws are taken to trivial ones. Thus we have a well-defined
action~$\Ld_S\colon \cl\mathcal{E}\to\cl\mathcal{E}$.

\subsection*{Coverings}
\label{sec:coverings}

Let us now give the main definition. Consider a locally trivial vector
bundle~$\tau\colon\tilde{\mathcal{E}}\to\mathcal{E}$ of rank~$r$ and denote
by~$\mathcal{F}(\mathcal{E})$ and~$\mathcal{F}(\tilde{\mathcal{E}})$ the
algebras of smooth functions on~$\mathcal{E}$ and~$\tilde{\mathcal{E}}$,
respectively. We have the
embedding~$\tau^*\colon\mathcal{F}(\mathcal{E})\hookrightarrow
\mathcal{F}(\tilde{\mathcal{E}})$.

\begin{definition}\label{def:basic-notions-1}
  We say that~$\tau$ carries a \emph{covering structure} (or is a
  \emph{differential covering} over~$\mathcal{E}$) if: (a) there exists a flat
  connection~$\tilde{\mathcal{C}}$ in the
  bundle~$\pi_\infty\circ\tau\colon\tilde{\mathcal{E}} \to M$ and (b) this
  connection enjoys the equation
  \begin{equation*}
    \eval{\tilde{\mathcal{C}}_Z}_{\mathcal{F}(\mathcal{E})}=\mathcal{C}_Z
  \end{equation*}
  for all vector fields~$Z$ on~$M$.
\end{definition}

In local coordinates, any covering is determined by a system of vector fields
\begin{equation}\label{eq:2}
  \tilde{D}_{x^i}=D_{x^i}+X_i,\qquad i=1,\dots,n,
\end{equation}
on~$\tilde{\mathcal{E}}$, where~$X_i$ are $\tau$-vertical fields that satisfy
the relations
\begin{equation}\label{eq:3}
  D_{x^i}(X_j)-D_{x^j}(X_i)+[X_i,X_j]=0,\qquad 1\le i<j\le n.
\end{equation}
Let~$w^1,\dots,w^r$ be local coordinates in the fiber of~$\tau$ (the
\emph{nonlocal variables} in~$\tau$) and~$X_i=X_i^1\pd{}{w^1}+\dots
+X_i^r\pd{}{w^r}$. Then~$\tilde{\mathcal{E}}$, endowed
with~$\tilde{\mathcal{C}}$, is equivalent to the overdetermined system of PDEs
\begin{equation}\label{eq:7}
  \pd{w^\alpha}{x^i}=X_i^\alpha,\qquad i=1,\dots,n,\quad\alpha=1,\dots,r,
\end{equation}
compatible by virtue of~$\mathcal{E}$.

Two coverings~$\tau_i\colon\tilde{\mathcal{E}}_i\to\mathcal{E}$, $i=1$, $2$,
are \emph{equivalent} if there exists a
diffeomorphism~$f\colon\tilde{\mathcal{E}}_1\to \tilde{\mathcal{E}}_2$ such
that the diagram
\begin{equation*}
  \xymatrix{
    \tilde{\mathcal{E}}_1\ar[rr]^f\ar[dr]_{\tau_1}&
    &\tilde{\mathcal{E}}_2\ar[dl]^{\tau_2}\\
    &\mathcal{E}&}
\end{equation*}
is commutative and~$f_*\circ\tilde{\mathcal{C}}_Z^1=\tilde{\mathcal{C}}_Z^2$
for all fields~$Z$ on~$M$, where~$\tilde{\mathcal{C}}^i$ is the Cartan
connection on~$\tilde{\mathcal{E}}_i$ and~$f_*$ is the differential of~$f$.

Again, having two coverings~$\tau_1$ and~$\tau_2$, consider the Whitney
product of fiber bundles
\begin{equation*}
  \xymatrix{
    &\tilde{\mathcal{E}}_1\times_{\mathcal{E}}\tilde{\mathcal{E}}_2
    \ar[dd]|{\tau_1\times\tau_2} \ar[dl]_{\tau_1^*(\tau_2)}
    \ar[dr]^{\tau_2^*(\tau_1)}&\\
    \tilde{\mathcal{E}}_1\ar[dr]_{\tau_1}&&\tilde{\mathcal{E}}_2\ar[dl]^{\tau_2}\\
    &\mathcal{E}\rlap{.}&}
\end{equation*}
Since the tangent plane
to~$\tilde{\mathcal{E}}_1\times_{\mathcal{E}}\tilde{\mathcal{E}}_2$ at any
point splits naturally into direct sum of tangent planes
to~$\tilde{\mathcal{E}}_1$ and~$\tilde{\mathcal{E}}_2$, we can define a
connection in the bundle~$\tau_1\times\tau_2$ by setting
\begin{equation*}
  \tilde{\mathcal{C}}_Z^{12}(\phi_1\phi_2)=
  \tilde{\mathcal{C}}_Z^1(\phi_1)\cdot\phi_2 
  +\phi_1\cdot\tilde{\mathcal{C}}_Z^2(\phi_2),\qquad
  \phi_1\in\mathcal{F}(\tilde{\mathcal{E}}_1),
  \quad\phi_2\in\mathcal{F}(\tilde{\mathcal{E}}_2). 
\end{equation*}
This is a covering structure in~$\tau_1\times\tau_2$ which is called the
\emph{Whitney product} of~$\tau_1$ and~$\tau_2$. Note that the
maps~$\tau_1^*(\tau_2)$ and~$\tau_2^*(\tau_1)$ are coverings as well (they are
called \emph{pull-backs}).

Assume that in local coordinates the coverings~$\tau_1$ and~$\tau_2$ are given
by the vector fields
\begin{equation*}
  \tilde{D}_{x^i}^k=D_{x^i} +
  \sum_{\alpha=1}^{\dim\tau_k}X_i^{k,\alpha}\pd{}{w_k^\alpha},\qquad k=1,2.
\end{equation*}
Then the vector fields defining the Whitney product are of the form
\begin{equation*}
  \tilde{D}_{x^i}^{1,2}=D_{x^i} +
  \sum_{\alpha=1}^{\dim\tau_1}X_i^{1,\alpha}\pd{}{w_1^\alpha} +
  \sum_{\alpha=1}^{\dim\tau_2}X_i^{2,\alpha}\pd{}{w_2^\alpha}.
\end{equation*}

\begin{remark}\label{sec:coverings-1}
  From now on we shall assume all the coverings under consideration to be
  finite-dimensional. It is known that, see~\cite{Igonin-UMN,Marvan}, in the
  multi-dimensional case (i.e.,~$n>2$) non-overdetermined equations do not
  possess finite-dimensional coverings. So, we restrict ourselves to the
  case~$n=2$.
\end{remark}

\begin{definition}
  A covering~$\tau\colon\tilde{\mathcal{E}}\to\mathcal{E}$ is called
  \emph{irreducible} if the covering equation~$\tilde{\mathcal{E}}$ is
  differentially connected, i.e., if for any set~$Z_1,\dots,Z_n$ of
  independent vector fields on~$M$ the system
  \begin{equation}\label{eq:1}
    \tilde{\mathcal{C}}_{Z_i}(h)=0,\qquad i=1,\dots,n,
  \end{equation}
  possesses constant solutions only. Otherwise we say that~$\tau$ is
  \emph{reducible}.
\end{definition}

Equivalently, one can study the equations
\begin{equation*}
  \tilde{D}_{x^i}(h)=0,\qquad i=1,\dots,n,
\end{equation*}
instead of System~\eqref{eq:1}.

Reducibility can be `measured' by the maximal number of functionally
independent integrals of Equation~\eqref{eq:1}, which cannot
exceed~$r=\dim\tau$. Maximally reducible coverings are called
\emph{trivial}. Triviality of a covering means that it is locally equivalent
to the one with~$\tilde{D}_{x^i}=D_{x^i}$ for all~$i=1,\dots,n$. Also,
directly from the definition one has

\begin{proposition}
  Any finite-dimensional covering~$\tau$\textup{,} in a neighborhood of a
  generic point\textup{,} splits into the Whitney
  product~$\tau=\tau_{\mathrm{triv}}\times\tau_{\mathrm{irr}}$\textup{,}
  where~$\tau_{\mathrm{triv}}$ is trivial and~$\tau_{\mathrm{irr}}$ is
  irreducible.
\end{proposition}

\subsection*{Abelian coverings}
\label{sec:abelian-coverings}
Let us introduce an important class of coverings.

\begin{definition}
  A covering~$\tau$ is called \emph{Abelian} if for any vector field~$X$
  on~$M$ one has~$\tilde{\mathcal{C}}_X(f)\in\mathcal{F}(\mathcal{E})$ for any
  fiber-wise linear function~$f\in\mathcal{F}(\tilde{\mathcal{E}})$. Coverings
  locally equivalent to such ones are also called Abelian.
\end{definition}

Locally this means that coordinates in the fibers of~$\tau$ may be chosen in
such a way that coefficient of the vertical fields~$X_i$ in
Equations~\eqref{eq:2} are independent of the nonlocal
variables~$w^\alpha$. Consider such a choice and recall that we are in the
two-dimensional situation (see Remark~\ref{sec:coverings-1}). Set~$x^1=x$,
$x^2=y$ and~$X_1=X$, $X_2=Y$. Then, by Equations~\eqref{eq:3}, one has
\begin{equation*}
  D_x(Y)-D_y(X)+[X,Y] =
  \sum_{\alpha=1}^r\big(D_x(Y^\alpha)-D_y(X^\alpha)\big)\pd{}{w^\alpha},
\end{equation*}
since~$[X,Y]=0$. Thus,~$D_x(Y^\alpha)-D_y(X^\alpha)=0$ for all~$\alpha$ and
all the forms
\begin{equation}\label{eq:4}
  \omega^\alpha=X^\alpha\d x+Y^\alpha\d y,\qquad\alpha=1,\dots,r,
\end{equation}
are conservation laws of~$\mathcal{E}$.

Using the result of this simple computation, let us give a complete
description of finite-dimensional Abelian coverings:

\begin{theorem}\label{sec:abelian-coverings-1}
  There locally exists a one-to-one correspondence between equivalence classes
  of $r$-dimensional irreducible Abelian coverings over~$\mathcal{E}$ and
  $r$-dimensional vector $\mathbb{R}$-subspaces in~$\cl\mathcal{E}$.
\end{theorem}

Let~$\tau$ be an Abelian covering of finite dimension~$r$. Then we can
construct conservation laws~$\{\omega^\alpha\}$ like in Equation~\eqref{eq:4}
and consider the space~$\mathcal{L}_\tau\subset\cl\mathcal{E}$ that spans the
set~$\{[\omega^\alpha]\}$ of their equivalence classes. Vice versa,
let~$\mathcal{L}\subset\cl\mathcal{E}$ be an $r$-dimensional subspace. Take
its basis~$e^1,\dots,e^r$ and choose a
representative~$\omega^\alpha=X^\alpha\d x+Y^\alpha\d y$ in each
class~$e^\alpha$. Consider~$\mathcal{E}\times\mathbb{R}^r$ and set
\begin{equation*}
  \tilde{D}_x=D_x+\sum_\alpha X^\alpha\pd{}{w^\alpha},\quad
  \tilde{D}_y=D_y+\sum_\alpha Y^\alpha\pd{}{w^\alpha}.
\end{equation*}
This obviously defines a covering structure and we denote it
by~$\tau_{\mathcal{L}}$. We are to prove that
\begin{enumerate}
\item\label{item:1} if~$\tau$ is irreducible of rank~$r$
  then~$\dim\mathcal{L}_\tau=r$;
\item\label{item:2} $\tau_{\mathcal{L}}$ is irreducible of rank~$r$;
\item\label{item:3} if~$\tau$ and~$\bar{\tau}$ are equivalent
  then~$\mathcal{L}_\tau=\mathcal{L}_{\bar{\tau}}$;
\item\label{item:4} equivalence class of~$\tau_{\mathcal{L}}$ is independent
  of a basis choice in~$\mathcal{L}$.
\end{enumerate}

\begin{proof}[Proof of Theorem~\textup{\ref{sec:abelian-coverings-1}}]
  Let us do it.

  \textbf{(\ref{item:1})} Take the forms~\eqref{eq:4} and assume that
  \begin{equation*}
    \lambda_1[\omega]^1+\dots+\lambda^r[\omega]^r=0
  \end{equation*}
  for some nontrivial set of~$\lambda^\alpha\in\mathbb{R}$. This means that
  \begin{equation*}
    \lambda_1\omega^1+\dots+\lambda^r\omega^r=\d_hP,\qquad
    P\in\mathcal{F}(\mathcal{E}),
  \end{equation*}
  or
  \begin{equation*}
    \sum_\alpha\lambda_\alpha X^\alpha=D_xP,\quad \sum_\alpha\lambda_\alpha
    Y^\alpha=D_yP. 
  \end{equation*}
  Hence, the function
  \begin{equation*}
    h=\lambda_1w^1+\dots+\lambda_rw^r-P
  \end{equation*}
  is a nontrivial integral for the fields~$\tilde{D}_x$, $\tilde{D}_y$ and the
  covering is reducible. Contradiction.

  \textbf{(\ref{item:2})} Consider the system
  \begin{equation*}
    \tilde{D}_x(h)=D_x(h)+\sum_\alpha X^\alpha\pd{h}{w^\alpha}=0,\quad
    \tilde{D}_y(h)=D_y(h)+\sum_\alpha Y^\alpha\pd{h}{w^\alpha}=0.
  \end{equation*}
  Note that when the partial derivatives~$\pd{h}{w^\alpha}$ vanish everywhere,
  the function~$h$ is constant, since~$\mathcal{E}$ is differentially
  connected.

  On the other hand, assume that there exists a
  point~$\theta\in\tilde{\mathcal{E}}$ such that at least one derivative,
  say~$\eval{\pd{h}{w^1}}_{\theta}\neq 0$. Hence, in a neighbor of~$\theta$ we
  can choose a new fiber coordinate~$\bar{w}^1=h$ and immediately make sure
  that the conservation law~$\omega^1$ is trivial.

  \textbf{(\ref{item:3})} Let~$\tau$ and~$\bar{\tau}$ be two equivalent
  irreducible Abelian coverings and
  \begin{equation*}
    \bar{w}^\alpha=f^\alpha(\theta,w^1,\dots,w^r),\qquad\theta\in\mathcal{E},
    \quad\alpha = 1,\dots,r,
  \end{equation*}
  be their equivalence. Then
  \begin{equation*}
    \bar{X}^\alpha=\tilde{D}_x(f^\alpha),\quad\bar{Y}^\alpha=\tilde{D}_y(f^\alpha),
    \qquad \alpha=1,\dots,r,
  \end{equation*}
  from where it follows that
  \begin{equation*}
    \tilde{D}_x\left(\pd{f^\alpha}{w^\beta}\right)=0,\qquad
    \tilde{D}_y\left(\pd{f^\alpha}{w^\beta}\right)=0
  \end{equation*}
  for all~$\alpha$, $\beta=1,\dots,r$. But the coverings are irreducible and
  consequently
  \begin{equation}\label{eq:5}
    f^\alpha=\sum_{\beta} a_\beta^\alpha w^\beta+a^\alpha,\qquad
    a_\beta^\alpha\in\mathbb{R}, \quad a^\alpha\in\mathcal{F}(\mathcal{E}),
  \end{equation}
  where~$\det{a_\beta^\alpha}\neq 0$. Thus
  \begin{equation*}
    \bar{X}^\alpha=\sum_{\beta} a_\beta^\alpha X^\beta+a^\alpha,
  \end{equation*}
  which means that~$\mathcal{L}(\tau)=\mathcal{L}(\bar{{\tau}})$.
  
  \textbf{(\ref{item:4})}
  Let~$\{[\omega^\alpha]\}$ and~$\{[\bar{\omega}^\beta]\}$ be two bases
  in~$\mathcal{L}$. Then
  \begin{equation*}
    \bar{\omega}^\alpha=\sum_{\beta}a_\beta^\alpha\omega^\beta+\d_h a^\alpha
  \end{equation*}
  and~\eqref{eq:5} is the needed equivalence.
\end{proof}

\section{The main result (Abelian case)}
\label{sec:main-result}
Using the above results, we study here some relations between symmetries and
conservation laws of the equations~$\mathcal{E}$ and~$\tilde{\mathcal{E}}$ in
an irreducible Abelian covering~$\tau\colon\tilde{\mathcal{E}}\to\mathcal{E}$.

\subsection*{Lifting conservation laws}
\label{sec:lift-cons-laws}
Consider a nontrivial conservation law~$\omega$ of~$\mathcal{E}$. Then the
pull-back~$\tau^*\omega$ is a conservation law of~$\tilde{\mathcal{E}}$. 

\begin{proposition}
  The form~$\tau^*\omega$ is a trivial conservation law of the
  equation~$\tilde{\mathcal{E}}$ if and only if~$[\omega]\in\mathcal{L}_\tau$.
\end{proposition}

\begin{proof}
  Consider the covering~$\tau\times\tau_{\omega}$, where~$\tau_{\omega}$ is
  the one-dimensional covering associated with the conservation
  law~$\omega$. Obviously, triviality of~$\tau^*\omega$ amounts to
  reducibility of~$\tau\times\tau_{\omega}$. But by
  Theorem~\ref{sec:abelian-coverings-1}, the
  covering~$\tau\times\tau_{\omega}$ is irreducible if and only
  if~$[\omega]\notin\mathcal{L}_\tau$.
\end{proof}

\begin{corollary}\label{sec:lift-cons-laws-1}
  Let~$\tau\colon\tilde{\mathcal{E}}\to\mathcal{E}$ be a finite-dimensional
  Abelian covering and assume
  that~$\dim_{\mathbb{R}}\cl\mathcal{E}=\infty$. Then~$\dim_{\mathbb{R}}
  \cl\tilde{\mathcal{E}} =\infty$ as well.
\end{corollary}

\begin{example}
  Consider the Korteweg-de~Vries equation
  \begin{equation*}
    u_t=uu_x+u_{xxx}
  \end{equation*}
  and its first conservation law~$\omega^1=u\d
  x+\left(\frac{u^2}{2}+u_{xx}\right)\d t$. In the corresponding
  covering~$\tau\colon\tilde{\mathcal{E}}\to\mathcal{E}$, the covering
  equation~$\tilde{\mathcal{E}}$ is the potential KdV
  \begin{equation*}
    u_t=\frac{1}{2}u_x^2+u_{xxx}.
  \end{equation*}
  All conservation laws of the KdV survive in pKdV except for~$\omega^1$.
\end{example}

\subsection*{Lifting symmetries}
\label{sec:lifting-symmetries}
Let~$\tau\colon\tilde{\mathcal{E}}\to\mathcal{E}$ be an arbitrary covering
and~$S\in\sym\mathcal{E}$ be a symmetry of the equation~$\mathcal{E}$. We say
that~$S$ \emph{lifts} to~$\tilde{\mathcal{E}}$ if there exists a
symmetry~$\tilde{S}\in\sym\tilde{\mathcal{E}}$ such
that~$\eval{\tilde{S}}_{\mathcal{F}(\mathcal{E})
  \subset\mathcal{F}(\tilde{\mathcal{E}})}=S$.

For any conservation law~$\omega$ of~$\mathcal{E}$, the Lie
derivative~$\Ld_S(\omega)$ is a conservation law as well and if two
conservation laws are equivalent then their Lie derivatives are also
equivalent. So, the action
\begin{equation}\label{eq:6}
\Ld_S\colon\cl\mathcal{E}\to\cl\mathcal{E}
\end{equation}
is well defined.

\begin{proposition}\label{sec:lifting-symmetries-2}
  Action~\eqref{eq:6} is $\mathbb{R}$-linear.
\end{proposition}

\begin{proof}
  Choose a basis in~$\cl\mathcal{E}$ and let~$\omega^\alpha=X^\alpha\d
  x+Y^\alpha\d y$, $\alpha=1,2,\dots$, be the corresponding conservation
  laws. Then, by~\eqref{eq:7}, we have nonlocal variables defined by the
  equations
  \begin{equation*}
    \pd{w^\alpha}{x}=X^\alpha,\qquad\pd{w^\alpha}{y}=Y^\alpha
  \end{equation*}
  for all possible values of~$\alpha$. Consequently,
  \begin{equation}\label{eq:8}
    S(X^\alpha)=\tilde{D}_x(\tilde{S}(w^\alpha)),\qquad
    S(Y^\alpha)=\tilde{D}_y(\tilde{S}(w^\alpha)),
  \end{equation}
  where
  \begin{equation*}
    \tilde{D}_x=D_x+\sum_\alpha X^\alpha\pd{}{w^\alpha},\qquad
    \tilde{D}_y=D_x+\sum_\alpha Y^\alpha\pd{}{w^\alpha},
  \end{equation*}
  The right-hand sides in~\eqref{eq:8} are independent of~$w^\beta$ while
  \begin{equation*}
    [\pd{}{w^\beta},\tilde{D}_x]=[\pd{}{w^\beta},\tilde{D}_y]=0
  \end{equation*}
  for all~$\beta$. Hence,
  \begin{equation*}
    \tilde{D}_x\left(\pd{\tilde{S}(w^\alpha)}{w^\beta}\right)=0,\qquad
    \tilde{D}_y\left(\pd{\tilde{S}(w^\alpha)}{w^\beta}\right)=0,
  \end{equation*}
  from where it follows (Theorem~\ref{sec:abelian-coverings-1}) that
  \begin{equation*}
    \pd{\tilde{S}(w^\alpha)}{w^\beta}=a_\beta^\alpha\in\mathbb{R},
  \end{equation*}
  or
  \begin{equation}\label{eq:9}
    \tilde{S}(w^\alpha)=\sum_\beta a_\beta^\alpha w^\beta+a^\alpha,\qquad
    a^\alpha\in\mathcal{F}(\mathcal{E})
  \end{equation}  
  (the sum above is taken over finite number of~$\beta$'s).
\end{proof}

\begin{remark}
  Equations~\eqref{eq:9} mean that~$\Ld_S\omega^\alpha=\sum_\beta
  a_\beta^\alpha\omega^\beta+\d_h a^\alpha$.
\end{remark}

An immediate consequence of this result is

\begin{proposition}\label{sec:lifting-symmetries-1}
  Let~$\mathcal{L}\subset\cl\mathcal{E}$ be an $r$-dimensional subspace
  and~$\tau_{\mathcal{L}}$ be the corresponding irreducible Abelian
  covering. Then a symmetry~$S$ lifts to~$\tau_{\mathcal{L}}$ if and only
  if~$\Ld_S(\mathcal{L})\subset\mathcal{L}$.
\end{proposition}

Assume now that~$\mathcal{E}$ admits an infinite-dimensional symmetry
algebra~$\mathcal{S}=\sym\mathcal{E}$. Consider a finite-dimensional
irreducible Abelian covering~$\tau\colon\tilde{\mathcal{E}}\to\mathcal{E}$
associated to conservation laws~$\omega^1,\dots,\omega^r$ and the
subspace~$\mathcal{L}_\tau\subset\cl\mathcal{E}$. Then, by
Proposition~\ref{sec:lifting-symmetries-1},~$\mathcal{S}$ lifts to the
covering~$\mathcal{S}\tau$ associated to the
space~$\Ld_{\mathcal{S}}\mathcal{L}$ that spans all the the conservation
laws
\begin{equation}\label{eq:13}
  \omega_0^\alpha=\omega^\alpha,\ \omega_1^\alpha=\Ld_S\omega^\alpha,\
  \omega_2^\alpha=\Ld_S\omega_1^\alpha,\dots,\qquad S\in\mathcal{S},
  \alpha=1,\dots,r. 
\end{equation}
If the space~$\Ld_{\mathcal{S}}\mathcal{L}$ is finite-dimensional a stronger
result is valid which is based on the following

\begin{proposition}\label{sec:lifting-symmetries-3}
  Let
  \begin{enumerate}
  \item $\mathcal{E}$ possess an infinite dimensional symmetry
    algebra~$\mathcal{S}$\textup{;}
  \item $\tau$ be a finite-dimensional irreducible Abelian covering
    over~$\mathcal{E}$\textup{;}
  \item $\bar{\tau}$ be another finite-dimensional Abelian covering over the
    equation~$\mathcal{E}$ such that~$\tau\times\bar{\tau}$ is irreducible and
    any~$S\in\mathcal{S}$ lifts to~$\tau\times\bar{\tau}$.
  \end{enumerate}
  Then the exists infinite number of symmetries in~$\mathcal{S}$ that lift
  to~$\tau$.
\end{proposition}

\begin{proof}
  Due to Proposition~\ref{sec:lifting-symmetries-2}, this is a fact from
  linear algebra. Choose bases~$w^1,\dots,w^r$
  and~$\bar{w}^1,\dots,\bar{w}^{\bar{r}}$ in~$\mathcal{L}_\tau$
  and~$\mathcal{L}_{\bar{\tau}}$, respectively. Then
  \begin{equation*}
    S_k(w^\alpha)=\sum_{\beta=1}^r\lambda_{k\beta}^\alpha w^\beta +
    \sum_{\bar{\beta}=1}^{\bar{r}} \bar{\lambda}_{k\bar{\beta}}^\alpha
    \bar{w}^{\bar{\beta}},
  \end{equation*}
  for all~$\alpha=1,\dots,r$ and~$k\ge 1$. Thus, the action of~$S_k$
  on~$\mathcal{L}_\tau$ is determined by two matrices
  \begin{equation*}
    \Lambda_k=
    \begin{pmatrix}
      \lambda_{k\beta}^\alpha
    \end{pmatrix}_{\beta=1,\dots,r}^{\alpha=1,\dots,r},\qquad
    \bar{\Lambda}_k=
    \begin{pmatrix}
      \bar{\lambda}_{k\bar{\beta}}^\alpha
    \end{pmatrix}_{\bar{\beta}=1,\dots,\bar{r}}^{\alpha=1,\dots,r}.
  \end{equation*}
  Then the space that spans the matrices~$\bar{\Lambda}_k$ is of
  dimension~$d\le r\cdot\bar{r}$ at most. Choose its
  basis~$\bar{\Lambda}_{k_1},\dots,\bar{\Lambda}_{k_d}$. Then
  \begin{equation*}
    \bar{\Lambda}_k=
    \mu_k^1\bar{\Lambda}_{k_1}+\dots+\mu_k^d\bar{\Lambda}_{k_d},\qquad k\ge 1,
  \end{equation*}
  and consequently the space~$\mathcal{L}_\tau$ is invariant with respect to
  all symmetries of the form~$\bar{S}_k=S_k-\sum_{i=1}^d\mu_k^iS_{k_i}$.
\end{proof}

Denote by~$\mathcal{S}\tau$ the space of all conservation laws generated
from~$\mathcal{L}_\tau$ by the iterated action of~$\mathcal{S}$, see
Equations~\eqref{eq:13}.

\begin{corollary}\label{sec:lifting-symmetries-4}
  Let~$\mathcal{S}\tau$ be finite-dimensional. Then there exists infinite
  number of independent symmetries in~$\mathcal{S}$ that lift to~$\tau$.
\end{corollary}

\begin{remark}
  Actually the proof of Proposition~\ref{sec:lifting-symmetries-3} shows that
  the symmetries~$\bar{S}_k$ may be chosen in such a way that they will act on
  all nonlocal variables trivially.
\end{remark}

\begin{example}
  Consider the Burgers equation
  \begin{equation*}
    u_t=uu_x+u_{xx}.
  \end{equation*}
  It possesses only one conservation law~$\omega^1=u\d x +
  \left(\frac{1}{2}u^2+u_x\right)\d t$ which has to be invariant with respect
  to all symmetries. Consequently, all these symmetries lift to the covering
  equation, which is the heat equation.
\end{example}

\subsection*{The main result}
\label{sec:main-result-1}
Gathering the results obtained above, we obtain the following
\begin{theorem}
  Let~$\mathcal{E}$ be a differentially connected equation
  and~$\tau\colon\tilde{\mathcal{E}}\to\mathcal{E}$ be an irreducible
  finite-dimensional Abelian covering. Then\textup{:}
  \begin{enumerate}
  \item\label{item:5} if~$\mathcal{E}$ possesses infinite number of
    conservation laws the same is valid for~$\tilde{\mathcal{E}}$\textup{;}
  \item\label{item:6} if~$\mathcal{E}$ possesses infinite number of symmetries
    then~$\tilde{\mathcal{E}}$ either has the same property or admits infinite
    number of conservation laws or both.
  \end{enumerate}
\end{theorem}

\begin{proof}
  Statement~(\ref{item:5}) is Corollary~\ref{sec:lift-cons-laws-1} exactly.

  To prove Statement~(\ref{item:6}), consider the space generated
  from~$\cl\mathcal{E}$ by~$\sym\mathcal{E}$. There are two options: (a) the
  space is finite-dimensional and we find ourselves in the situation of
  Corollary~\ref{sec:lifting-symmetries-4}; (b) otherwise we come back to
  Corollary~\ref{sec:lift-cons-laws-1}.
\end{proof}

\section{Non-Abelian case}
\label{sec:disc-non-abel}
The non-Abelian case is more complicated, and the first thing to be done is to
narrow the universum of non-Abelian coverings.
\begin{definition}
  A finite-dimensional covering~$\tau\colon\tilde{\mathcal{E}}\to\mathcal{E}$
  is called \emph{strictly non-Abelian} if it is not equivalent to a
  composition of coverings~$\xymatrix{\tilde{\mathcal{E}}\ar[r]^{\tau_1}&
    \mathcal{E}'\ar[r]^{\tau_2}& \mathcal{E}}$, where~$\tau_2$ is Abelian.
\end{definition}

\begin{proposition}\label{prop:non-abelian-case}
  Let~$\tau\colon\tilde{\mathcal{E}}\to\mathcal{E}$ be a finite-dimensional
  strictly non-Abelian covering and~$\omega$ be a nontrivial conservation law
  of the equation~$\mathcal{E}$. Then the conservation
  law~$\tau^*(\omega)\in\cl(\tilde{\mathcal{E}})$ is nontrivial as well.
\end{proposition}

\begin{proof}
  Let~$\omega=X\d x+Y\d y$, where~$X$ and~$Y$ are functions
  on~$\mathcal{E}$. Assume that~$\tau^*(\omega)$ is trivial. Then there exists
  a function~$f$ on~$\tilde{\mathcal{E}}$ such that
  \begin{equation*}
    X=\tilde{D}_x(f),\qquad Y=\tilde{D}_y(f),
  \end{equation*}
  where~$\tilde{D}_x$, $\tilde{D}_y$ are the total derivatives
  on~$\tilde{\mathcal{E}}$. Since~$\omega$ is nontrivial, at least one of the
  partial derivatives~$\pd{f}{w^i}$, say~$\pd{f}{w^1}$, does not vanish,
  where~$w^1,\dots,w^r$ are nonlocal variables in~$\tau$. Then, by choosing
  new coordinates
  \begin{equation*}
    \bar{w}^1=f,\ \bar{w}^2=w^2,\dots,\bar{w}^r=w^r
  \end{equation*}
  in the fiber, we see that~$\tau$ is not strictly non-Abelian.
\end{proof}

\begin{corollary}
  One has~$\dim_{\mathbb{R}}\ker\tau^*<\infty$ for any finite-dimensional
  covering~$\tau$.
\end{corollary}

\begin{proof}
  The result follows from the proof of Proposition~\ref{prop:non-abelian-case}.
\end{proof}

\begin{corollary}
  If~$\tau\colon\tilde{\mathcal{E}}\to\mathcal{E}$ is a strictly non-Abelian
  covering then the
  map~$\tau^*\colon\cl(\mathcal{E})\to\cl(\tilde{\mathcal{E}})$ is an
  embedding. In particular\textup{,} if~$\dim\cl(\mathcal{E})=\infty$ the same
  holds for~$\cl(\tilde{\mathcal{E}})$.
\end{corollary}

Consider now a symmetry~$S\in\sym\mathcal{E}$. Let try to formulate an
analogue of Proposition~\ref{sec:lifting-symmetries-1} in the non-Abelian
case. First of all, recall an old result from~\cite{VinKrasTrends}:

\begin{proposition}\label{prop:non-abelian-case-1}
  Let~$\tau\colon\tilde{\mathcal{E}}\to\mathcal{E}$ be a finite-dimensional
  covering and~$S\in\sym\mathcal{E}$ be a symmetry that possesses a
  one-parameter group~$A_\lambda$ of transformations \textup{(}e.g.\textup{,}
  a contact symmetry\textup{)}. Then\textup{:}
  \begin{itemize}
  \item either~$S$ can be lifted to a
    symmetry~$\tilde{S}\in\sym\tilde{\mathcal{E}}$ that is projectible to~$S$
    by~$\tau_*$\textup{,}
  \item or the action of the group~$A_\lambda$ gives rise to a one-parameter
    family of
    coverings~$\tau_\lambda\colon\tilde{\mathcal{E}}\to\mathcal{E}$\textup{,}
    $\tau_0=\tau$\textup{,} with a nonremovable
    parameter~$\lambda\in\mathbb{R}$.
  \end{itemize}
\end{proposition}

To formulate a counterpart to Proposition~\ref{prop:non-abelian-case-1} in the
case when~$A_\lambda$ does not exist, let us recall the basic constructions
from~\cite{KrasSomeNew}. Let~$\mathcal{E}\subset J^\infty(\pi)$ be an equation
with a set of internal coordinates~$\ic(\mathcal{E})$. The \emph{structural
  element} of~$\mathcal{E}$ is the vector-valued differential one-form
\begin{equation}
  \label{eq:10}
  U_{\mathcal{E}}=
  \sum_{u_\sigma^j\in\ic(\mathcal{E})}\omega_\sigma^j\otimes\pd{}{u_\sigma^j},
\end{equation}
where~$\omega_\sigma^j=\d u_\sigma^j-\sum_iu_{\sigma i}\d x^i$ are the
\emph{Cartan forms} corresponding to~$u_\sigma^j$. Denote
by~$D^v(\Lambda^*(\mathcal{E}))=\bigoplus_i D^v(\Lambda^i(\mathcal{E}))$
differential forms on~$\mathcal{E}$ with values in~$\pi_\infty$-vertical
vector fields (or, in other words, vertical form-valued derivations of the
function
algebra~$\mathcal{F}(\mathcal{E})$). Then~$D^v(\Lambda^*(\mathcal{E}))$ is a
super Lie algebra with respect to the \emph{Nijenhuis bracket}
\begin{equation*}
\llbracket\cdot\,,\cdot\rrbracket\colon D^v(\Lambda^i(\mathcal{E}))
\times D^v(\Lambda^j(\mathcal{E}))\to D^v(\Lambda^{i+j}(\mathcal{E})).
\end{equation*}
The element~$U_{\mathcal{E}}$ from~$~\eqref{eq:10}$ can be understood as a
derivation~$U_{\mathcal{E}}\in
D_1^v(\Lambda^1(\mathcal{E}))$. Then~$\llbracket
U_{\mathcal{E}},U_{\mathcal{E}}\rrbracket=0$ and we have the complex
\begin{equation}
  \label{eq:11}
  \xymatrix{
    0\ar[r]&D^v(\mathcal{E})\ar[r]^-{\partial_{\mathcal{C}}}&
    D^v(\Lambda^1(\mathcal{E}))\ar[r]^-{\partial_{\mathcal{C}}}&
    \dots\ar[r]&D^v(\Lambda^i(\mathcal{E}))\ar[r]^-{\partial_{\mathcal{C}}}&\dots 
  }
\end{equation}
where~$\partial_{\mathcal{E}}=\llbracket
U_{\mathcal{E}},\cdot\rrbracket$. Its cohomology is called the
\emph{$\mathcal{C}$-cohomology} of~$\mathcal{E}$ and is denoted
by~$H_{\mathcal{C}}^i(\mathcal{E})$.

\begin{theorem}[see~\cite{KrasSomeNew}]\label{thm:non-abelian-case}
  Let~$\mathcal{E}\subset J^\infty(\pi)$ be an equation. Then\textup{:}
  \begin{enumerate}
  \item $H_{\mathcal{C}}^0(\mathcal{E})=\sym\mathcal{E}$\textup{,}
  \item $H_{\mathcal{C}}^1(\mathcal{E})$ consists of equivalence classes of
    infinitesimal deformations of the equation structure modulo trivial
    ones\textup{,}
  \item $H_{\mathcal{E}}^2(\mathcal{C})$ contains obstructions to continuation
    of infinitesimal deformations to formal ones.
  \end{enumerate}
\end{theorem}

Let~$\tau\colon\tilde{\mathcal{E}}\to\mathcal{E}$ be a covering with local
coordinates~$w^\alpha$ in its fibers. The the structural element
of~$\tilde{\mathcal{E}}$ is
\begin{equation}\label{eq:12}
  U_{\tilde{\mathcal{E}}}=U_{\mathcal{E}}+\sum_\alpha\theta^\alpha\otimes\pd{}{w^\alpha},
\end{equation}
where
\begin{equation*}
  \theta^\alpha=\d w^\alpha-\sum_i X_i^\alpha\d x^i
\end{equation*}
and the functions~$X_i^\alpha$ are from Equation~\eqref{eq:7}.

Denote the second summand in~\eqref{eq:12} by~$U_\tau$. Then the covering
structure in~$\tau$ is governed by the \emph{Maurer-Cartan type} equations
\begin{equation*}
  \partial_{\mathcal{C}}(U_\tau)+\frac{1}{2}\llbracket U_\tau,U_\tau\rrbracket=0.
\end{equation*}

Consider complex~\eqref{eq:11} for the covering equation~$\tilde{\mathcal{E}}$
and its subcomplex
\begin{equation*}
  \xymatrix{
    0\ar[r]&D^g(\tilde{\mathcal{E}})\ar[r]^-{\partial_{\mathcal{C}}}&
    D^g(\Lambda^1(\tilde{\mathcal{E}}))\ar[r]^-{\partial_{\mathcal{C}}}&
    \dots\ar[r]&D^g(\Lambda^i(\tilde{\mathcal{E}}))
    \ar[r]^-{\partial_{\mathcal{C}}}&\dots,  
  }
\end{equation*}
where
\begin{equation*}
  D^g(\Lambda^i(\mathcal{E}))=\sd{Z\in
    D^v(\Lambda^i(\mathcal{E}))}{\eval{Z}_{\mathcal{F(\mathcal{E})}}=0}
\end{equation*}
(recall that the algebra~$\mathcal{F(\mathcal{E})}$ is embedded
into~$\mathcal{F}(\tilde{\mathcal{E}})$ by~$\tau^*$). Denote the corresponding
cohomology groups by~$H_g^i(\tau)$. Then we have the following analogue of
Theorem~\ref{thm:non-abelian-case}:

\begin{theorem}\label{thm:non-abelian-case-1}
  Let~$\tau\colon\tilde{\mathcal{E}}\to\mathcal{E}$ be a
  covering. Then\textup{:}
  \begin{enumerate}
  \item $H_g^0(\tau)$ consists of gauge symmetries \textup{(}infinitesimal
    equivalences\textup{)} of~$\tau$\textup{,}
  \item $H_g^1(\tau)$ consists of equivalence classes of
    the covering structure infinitesimal deformations modulo trivial
    ones\textup{,}
  \item\label{item:7} $H_g^2(\tau)$ is the set of obstructions to continuation
    of infinitesimal deformations to formal ones.
  \end{enumerate}
\end{theorem}

To formulate the last result of this paper, let us give the following

\begin{definition}
  Let~$\mathcal{E}$ be an equation
  and~$\tau\colon\tilde{\mathcal{E}}\to\mathcal{E}$ be a trivial vector
  bundle. We say that~$\tau$ carries a \emph{formal covering
    structure}~$\tau_\epsilon$, where~$\epsilon$ is a formal parameter, if
  there exist formal $\tau$-vertical vector fields~$X_i=\sum_{i=0}^\infty
  \epsilon^s X_{is}$ such that the equalities
  \begin{equation*}
    D_{x^i}(X_j)-D_{x^j}(X_j)+[X_i,X_j]=0
  \end{equation*}
  hold formally for all~$i$, $j=1,\dots,n$.
\end{definition}

Let now~$S$ be a symmetry of~$\mathcal{E}$. Then (locally)~$S$ can be lifted
to a vector field~$\tilde{S}$ on~$\tilde{\mathcal{E}}$. Two options are
possible: (1)~$\tilde{S}$ is a symmetry of~$\tilde{\mathcal{E}}$;
(2)~$\tilde{S}$ is not a symmetry. In the latter case, consider the
element~$\Ld_{\tilde{S}}(U_{\tilde{\mathcal{E}}})= \llbracket
\tilde{S},U_{\tilde{\mathcal{E}}}\rrbracket$. One has

\begin{proposition}
  The element~$\Ld_{\tilde{S}}(U_{\tilde{\mathcal{E}}})$ is a $1$-cocycle
  in~\eqref{eq:11}.
\end{proposition}

\begin{proof}
  The result immediately follows from the two facts:

  (1) For any elements~$\Omega\in D^v(\Lambda^i)$ and a vertical vector
  field~$Z$ one has
  \begin{equation*}
    \big(\Ld_Z(\Omega)\big)(\phi) = \llbracket Z,\Omega\rrbracket(\phi) =
    \Ld_Z\big(\Omega(\phi)\big) - (-1)^i\Omega\big(Z(\phi)\big),
  \end{equation*}
  where~$\phi$ is an arbitrary smooth function on~$\tilde{\mathcal{E}}$.

  (2) $L_{\tilde{S}}(U_{\tilde{\mathcal{E}}}) =
  \llbracket\tilde{S},U_{\tilde{\mathcal{E}}}\rrbracket = \llbracket S ,
  U_\tau\rrbracket + \llbracket S_g,U_{\mathcal{E}} + U_\tau\rrbracket$,
  where~$S_g\in D^g(\tilde{\mathcal{E}})$ is locally defined by~$S_g=\tilde{S}
  -S$.
\end{proof}

Hence,
\begin{equation*}
  \bar{U}_\tau=U_\tau+\epsilon\cdot \Ld_{\tilde{S}}(U_{\tilde{\mathcal{E}}})
\end{equation*}
is an infinitesimal deformation of the covering structure in~$\tau$. This
deformation is trivial if and only if the lift~$\tilde{S}$ is a symmetry
of~$\tilde{\mathcal{E}}$.

\begin{theorem}
  Let~$\tau\colon\tilde{\mathcal{E}}\to\mathcal{E}$ be a covering such
  that~$H_g^2(\tau)=0$ and~$S\in\sym\mathcal{E}$. Then there exists a formal
  covering structure~$\tau_\epsilon$ in~$\tau$ such that~$\tau_0=\tau$.
\end{theorem}

\begin{proof}
  The result immediately follows from
  Theorem~\ref{thm:non-abelian-case-1}\,(\ref{item:7}).
\end{proof}

\begin{remark}
  This result is a weaker analogue of
  Proposition~\ref{prop:non-abelian-case-1}. It may have stronger consequences
  with additional assumptions. For example, if we assume that the vector
  fields~$X_i$ depend on the nonlocal variables polynomially then, expanding
  these variables in formal series with respect to the deformation
  parameter~$\epsilon$, we shall obtain an infinite-dimensional covering
  (genuine, not formal) over~$\mathcal{E}$.

  Moreover, if the covering equation~$\tilde{\mathcal{E}}_\epsilon$ is in the
  divergent form there is a hope to construct infinite number of conservation
  laws for~$\mathcal{E}$, similar to the classical construction applied to the
  Korteweg-de~Vries equation using the Miura transform.
\end{remark}

\section*{Acknowledgments}
\label{sec:acknowledgements}
I am grateful to Artur Sergyeyev who asked the question which is, hopefully,
answered in this paper. My thanks are also due to Alik Verbovetsky for
criticism and fruitful discussions.

\end{document}